%% file: main.tex
\documentclass[letterpaper, 10 pt, conference]{ieeeconf}  

\IEEEoverridecommandlockouts                              
                                                          
\usepackage[utf8]{inputenc}
\usepackage{include_packages} 

\title{\LARGE \bf When Smoothness is Not Enough:\\\mbox{Toward Exact Quantification and Optimization of the Price-of-Anarchy} \\
\thanks{This work is supported by ONR grant \#N00014-17-1-2060, NSF grant \#ECCS-1638214, and SNSF grant \#P2EZP2\_181618. 
}}
\author{Rahul Chandan \and Dario Paccagnan \and Jason R. Marden
\thanks{R. Chandan and J.R. Marden are with the Department of Electrical and Computer Engineering and the Center of Control, Dynamical Systems and Computation, UC Santa Barbara, USA. Email: \href{mail_to:rchandan@ucsb.edu}{rchandan@ucsb.edu}, \href{mail_to:jrmarden@ece.ucsb.edu}{jrmarden@ece.ucsb.edu}.}
\thanks{D. Paccagnan is with the Department of Mechanical Engineering and the Center of Control, Dynamical Systems and Computation, UC Santa Barbara, USA. Email: \href{mail_to:dariop@ucsb.edu}{dariop@ucsb.edu}.}}

\begin{document}
\renewcommand{\arraystretch}{1.2}
\renewcommand{\thetable}{\arabic{table}}

\maketitle 

\begin{abstract}
    Today's multiagent systems have grown too complex to rely on centralized controllers, prompting increasing interest in the design of distributed algorithms. 
    In this respect, game theory has emerged as a valuable tool to complement more traditional techniques.
    The fundamental idea behind this approach is the assignment of agents' local cost functions, such that their selfish minimization attains, \hbox{or is provably close to,} the global objective. 
    \hbox{Any algorithm capable of computing an} equilibrium of the corresponding \hbox{game inherits an approximation} ratio that is, in the worst case, equal to its price-of-anarchy.
    Therefore, a successful application of the game design approach hinges on the possibility to quantify and optimize the equilibrium performance.

    Toward this end, we introduce the notion of generalized smoothness, and show that the resulting efficiency bounds are significantly tighter compared to those obtained using the traditional smoothness approach. 
    Leveraging this newly-introduced notion, we quantify the equilibrium performance for the class of local resource allocation games.
    Finally, we show how the agents' local decision rules can be designed in order to optimize the efficiency of the corresponding equilibria, by means of a tractable linear program.
\end{abstract}

\input{introduction}
\input{problem_statement}
\input{smoothness}

\input{tightness}
\input{lp}
\input{extensions}

\section{Conclusions}
    In this manuscript, we provided a novel methodology for characterizing and optimizing the price-of-anarchy in connection to a broad class of problems, including congestion games. Toward this goal, we introduced the notion of generalized smoothness. Compared to traditional smoothness arguments, we showed that generalized smoothness is more widely applicable, and provides tighter bounds. We applied generalized smoothness arguments to the class of local resource allocation problems (which include congestion games) and observed that it provides tight bounds on the price-of-anarchy. Relative to this class of problems, we were able to compute and optimize the price-of-anarchy of coarse-correlated equilibria, by means of concrete and tractable linear programs. Along with other possible future research directions, this work paves the way for the design of optimal tolling schemes through the linear programming framework introduced in \cite{paccagnan2018distributed,chandan2019optimal}.

\bibliographystyle{IEEEtran}
\bibliography{references}
\null \vfill
\appendix
\input{appendix}

\end{document}

%% file: introduction.tex
\section{Introduction}
\label{sec:intro}
Interest in the field of multiagent systems' control has experienced rapid growth in recent years, as a variety of application domains have emerged \cite{bussmann2013multiagent,wu2017distributed}. The impact of recent advancements in multiagent control has been far-reaching, revolutionizing traditional industries such as transportation and power networks \cite{yu2016smart,spieser2014toward,le2015decentralized}, while also driving the development of novel technologies including robotic swarms and self-driving cars \cite{el2013distributed,zhang2016collaborative}. 

Modern multiagent systems must adhere to imposing constraints with regards to their spatial distribution, overall scale, privacy requirements and communication bandwidth.
As a consequence, the coordination of such systems does not allow for centralized decision making,
but instead requires the use of distributed protocols.
Ideally, a distributed algorithm will meet the system's requirements for scalability, communication bandwidth, and security, while achieving the desired global objective.

A well-established and fruitful approach to tackle this class of problems consists in the design of a centralized maximization algorithm, that is later distributed by leveraging the structure of the problem considered, e.g., \cite{nedic2009distributed,wei2013distributed}.
An alternative approach, termed \textit{game design}, has emerged in parallel as a valuable tool to complement the aforementioned design philosophy \cite{shamma2007cooperative}.
Instead of directly specifying the decision-making process, local cost functions are assigned to the system's agents such that their selfish minimization results in the achievement of the system-level objective.


The advantages of using this approach are two-fold: i) we inherit a pool of algorithms that are distributed by nature, asynchronous, and resilient to external disturbances \cite{arslan2007autonomous}; and, ii) we obtain access to readily-available performance certificates in the form of efficiency bounds.
In fact, any (distributed) algorithm capable of driving the system to an equilibrium configuration (e.g. pure Nash equilibrium, mixed Nash equilibrium, correlated equilibrium, etc.) will inherit an approximation ratio matching the corresponding worst-case equilibrium efficiency, called the \textit{price-of-anarchy}.
Motivated by the game-theoretic approach, we aim to develop novel techniques to quantify and minimize the price-of-anarchy in distributed systems.

\subsection{Related Works}

The characterization of the price-of-anarchy has received significant research interest, particularly in the context of atomic congestion games~\cite{christodoulou2005price,awerbuch2005price,aland2006exact,gairing2009covering}. Although the smoothness framework provides exact price-of-anarchy bounds for atomic congestion games~\cite{roughgarden2009intrinsic}, the derivation of these bounds still requires a considerable amount of analysis.
In stark contrast, we construct tractable linear programs for computing and optimizing the price-of-anarchy relative to the class of local resource allocation games, without any analysis. 
These linear programs extend the approach put-forward in~\cite{paccagnan2018distributed,chandan2019optimal}.

As smoothness arguments have proven useful when characterizing the performance of broad classes of equilibria \cite{caragiannis2015bounding}, they have also been applied to a variety of other problems, including learning \cite{foster2016learning}, and mechanism design \cite{syrgkanis2013composable}.
Unfortunately, as observed in \cite{paccagnan2018distributed} and proven later in this manuscript, traditional smoothness arguments find limited applicability in connection to \emph{design problems}. Generalized smoothness is tailored to resolve this weakness, while retaining all the strengths of the traditional smoothness approach.
This novel notion of smoothness is most similar to the style of argument used in \cite{gairing2009covering,ramaswamy2017impact} to quantify the price-of-anarchy for covering problems.

\subsection{Our Contributions} 
In this work, we introduce a broader notion of smoothness, referred to as generalized smoothness, which allows us to provide tighter bounds on the performance of coarse-correlated equilibria. To demonstrate the strength of this novel approach, we apply our result to local resource allocation problems, and show that the bounds are tight. 
In more detail, our contributions are as follows:
\begin{enumerate}
    \item We demonstrate that price-of-anarchy bounds obtained via smoothness arguments are not tight if the sum of players' local cost functions is not equal to the system cost (\cref{prop:strictlooseness}).
    \item We introduce the notion of generalized smoothness, and show that, in general, it provides tighter bounds on the price-of-anarchy compared to current smoothness approaches (\cref{thm:gsmooth,lem:gsmoothbetter}).
    \item For the class of local resource allocation problems, we show that generalized smoothness provides tight bounds on the price-of-anarchy (\cref{thm:worstcasegame}). As a consequence, we demonstrate how the price-of-anarchy can be characterized (\cref{thm:equivalencetolp}) and optimized (\cref{thm:optimaldistributionrules}) using tractable linear programs. 
    Finally, we show that many existing price-of-anarchy results, e.g., \cite{christodoulou2005price,awerbuch2005price,aland2006exact}, can be reproduced by simply solving a corresponding linear program.
\end{enumerate}

For ease of presentation, many of the proofs are reported in the appendix.

%% file: problem_statement.tex
\section{Problem Statement} \label{sec:problemstatement}

Consider a class of resource allocation problems where $N = \{1, \dots, n\}$ denotes a set of agents, and each agent $i$ must select an action $a_i$ from a given action set $\mathcal{A}_i$. The system cost induced by allocation $a = (a_1, \dots, a_n) \in \mathcal{A} = \mathcal{A}_1 \times \dots \times \mathcal{A}_n$ is $C(a)$, where $C:\mathcal{A} \to \mathbb{R}$. Our objective is to find an optimal allocation, i.e. an allocation $a^\mathrm{opt}$,
\begin{equation} \label{eq:optimal_allocation}
    a^{\mathrm{opt}} \in \argmin_{a \in \mathcal{A}} C(a)\text{.}
\end{equation}
Since this class of combinatorial problems is inherently intractable,
in the remainder of the paper, we aim to obtain an approximate solution to \eqref{eq:optimal_allocation} through a distributed and tractable algorithm, ideally with the best possible approximation ratio\footnote{%
For ease of presentation, most of our analysis will focus on Nash equilibria, which are intractable to find, or even nonexistent, in general. Nevertheless, we will show in \cref{thm:gsmoothcce} that our results generalize to the much broader set of coarse-correlated equilibria, which can be found in polynomial time for a broad class of games \cite{papadimitriou2008computing}, and are guaranteed to exist in a broader class of games than pure Nash equilibria \cite{roughgarden2015intrinsic}.}. 
We tackle the problem using the game-theoric approach discussed in the introduction. 
Towards this goal, for every instance of problem in \eqref{eq:optimal_allocation}, we introduce a corresponding game where the agent set is $N$, each agent's action set is $\mathcal{A}_i$, and in which each agent $i$ evaluates its actions using a local cost function $J_i : \mathcal{A} \to \mathbb{R}$. In the forthcoming analysis, we will focus on the solution concept of Nash equilibrium, defined as any allocation $a^\textrm{ne} \in \mathcal{A}$ such that, 
\begin{equation} \label{eq:equilibriumconditions}
    J_i(a^\textrm{ne}) \leq J_i(a_i, a_{-i}^\textrm{ne}) \quad \forall a_i \in \mathcal{A}_i, \forall i \in N,
\end{equation}
where $a_{-i} = (a_1, \dots, a_{i-1}, a_{i+1}, \dots, a_n)$.
We represent the game as defined above with the tuple $G = (N,\mathcal{A}, \{J_i\}, C)$, where $\{J_i\} = \{J_1, \dots, J_n\}$. We measure the equilibrium performance in a given game using the notion of price-of-anarchy,
\begin{equation}
    \textrm{PoA}(G) := \frac{\max_{a\in\textrm{NE}(G)} C(a)}{\min_{a \in \mathcal{A}} C(a)}\text{,}
\end{equation}
where $\textrm{NE}(G)$ is the set of all pure Nash equilibria of the game $G$. Informally, the price-of-anarchy describes the ratio between the worst performing equilibrium and the optimal allocation. 
A lower price-of-anarchy is indicative of higher overall equilibrium performance. As such, the price-of-anarchy is an upper-bound on the efficiency of \emph{any} equilibrium in the game. In cases where we have a family of games $\mathcal{G}$, the price-of-anarchy is further defined as,
\begin{equation}
    \textrm{PoA}(\mathcal{G}) := \sup_{G \in \mathcal{G}} \textrm{PoA}(G).
\end{equation}
Our work centers around the following two questions:
\begin{enumerate}
    \item Given a class of cost-minimization games, how do we quantify the price-of-anarchy?
    \item How can agents' local cost functions be designed in order to minimize the price-of-anarchy?
\end{enumerate}

\section{The smoothness framework}
The smoothness framework developed in \cite{roughgarden2015intrinsic} has proven to be versatile, bringing a number of different price-of-anarchy results under a common analytical language, and producing tight bounds on the price-of-anarchy for different classes of problems \cite{roughgarden2009intrinsic,roughgarden2017price}.
In this section, we revisit the notion of smooth games, and recall how smoothness arguments are employed to bound the corresponding price-of-anarchy. 
The cost-minimization game $G$ is ($\lambda,\mu$)-smooth if $\sum_{i=1}^n J_i(a) \geq C(a)$ for all $a \in \mathcal{A}$, and if, for any two allocations $a, a' \in \mathcal{A}$, there exist $\lambda > 0$ and $\mu < 1$ such that
\begin{equation} \label{eq:smoothnessconditions}
    \sum_{i=1}^n J_i(a_i', a_{-i}) \leq \lambda\,C(a') + \mu\,C(a).
\end{equation}
The price-of-anarchy of a ($\lambda,\mu$)-smooth game $G$ is upper-bounded as $\rm{PoA}(G) \leq \lambda/(1-\mu)$.

Observe that if all games in a class $\mathcal{G}$ can be shown to be ($\lambda,\mu$)-smooth, then $\rm{PoA}(\mathcal{G})$ is also upper-bounded by $\lambda/(1-\mu)$.
Accordingly, the best price-of-anarchy bound that can be derived using the smoothness framework, termed the \textit{robust price-of-anarchy} \cite{roughgarden2015intrinsic}, is given by,
\begin{equation} \label{eq:rpoa}
\rm{RPoA}(\mathcal{G}) := \inf_{\lambda > 0, \mu < 1} \left\{ \frac{\lambda}{1-\mu} \text{ s.t. \eqref{eq:smoothnessconditions} holds } \forall G \in \mathcal{G} \right\}.
\end{equation}
Note that, in general, $\textrm{PoA}(\mathcal{G}) \leq \textrm{RPoA}(\mathcal{G})$. Fortunately, $\textrm{PoA}(\mathcal{G}) = \textrm{RPoA}(\mathcal{G})$ for the well-studied class of congestion games, in which $\sum_{i=1}^n J_i(a) = C(a)$ for all $a \in \mathcal{A}$, see \cite{roughgarden2015intrinsic}.

However, smoothness arguments are not applicable to games where $\sum_{i=1}^n J_i(a) < C(a)$ even for only one $a \in \mathcal{A}$. 
Additionally, the robust price-of-anarchy does not provide a tight bound when $\sum_{i=1}^n J_i(a) > C(a)$ for all $a \in \mathcal{A}$, as we demonstrate in the following theorem.
\begin{theorem} \label{prop:strictlooseness}
    For a given game $G$, assume $\sum_{i=1}^n J_i(a) > C(a)$ holds for all $a \in \mathcal{A}$. Then,
    \begin{equation}
        \rm{RPoA}(G) > \rm{PoA}(G).
    \end{equation}
\end{theorem}
\begin{proof}
    By assumption, there must exist $\gamma > 1$ such that $\sum_{i=1}^n J_i(a) \geq \gamma C(a)$ for all $a \in \mathcal{A}$. Observe that, for $\lambda > 0$ and $\mu < 1$ as in \eqref{eq:smoothnessconditions},
    \begin{align*}
        \gamma C(a^\textrm{ne}) &\leq \sum_{i=1}^n J_i(a^\textrm{ne}) \leq \sum_{i=1}^n J_i(a^\textrm{opt}, a_{-i}^\textrm{ne}) \\
        &\leq \lambda\,C(a^\textrm{opt}) + \mu\,C(a^\textrm{ne}), 
    \end{align*}
    where the above inequalities hold by assumption, by \eqref{eq:equilibriumconditions}, and by \eqref{eq:smoothnessconditions}, respectively. As the equilibrium conditions in \eqref{eq:equilibriumconditions} are scale-invariant, it must be that
    \[ \rm{PoA}(G) \leq \frac{\lambda^*}{\gamma - \mu^*} < \frac{\lambda^*}{1 - \mu^*} = \textrm{RPoA}(G), \]
    where $\lambda^* > 0$, $\mu^* < 1$ optimize \eqref{eq:rpoa}.
\end{proof} 

%% file: smoothness.tex
\section{Generalized smoothness}
In the previous section, we showed that traditional smoothness arguments are unsuitable for bounding equilibrium performance when the sum of agents' local costs is not equal to the system cost. In the following, we introduce a new notion of smoothness that provides tight bounds for a broader class of games.

\begin{theorem}[Generalized Smoothness] \label{thm:gsmooth}
    Suppose there exist ${\lambda} > 0$ and ${\mu} < 1$ such that for every game $G \in \mathcal{G}$, and any two action profiles $a, a' \in \mathcal{A}$,
    \begin{equation} \label{eq:gsmooth}
        \sum_{i=1}^n J_i(a_i', a_{-i}) - \sum_{i=1}^n J_i(a) + C(a) \leq {\lambda}\,C(a') + {\mu}\,C(a).
    \end{equation}
    Then, the price-of-anarchy satisfies,
    \[ \rm{PoA}(\mathcal{G}) \leq \frac{{\lambda}}{1-{\mu}}. \]
\end{theorem}
\begin{proof}
    For all $G \in \mathcal{G}$, for all $a^\textrm{ne} \in \textrm{NE}(G)$ and $a^\textrm{opt} \in \mathcal{A}$,
    \begin{equation}\label{eq:gsmoothinequality}
        \begin{split}
            C(a) &\leq \sum_{i=1}^n J_i(a_i^\textrm{opt}, a_{-i}^\textrm{ne}) - \sum_{i=1}^n J_i(a^\textrm{ne}) + C(a) \\
            &\leq {\lambda}\,C(a^\textrm{opt}) + {\mu}\,C(a^\textrm{ne}),
        \end{split}
    \end{equation}
    where the first inequality holds by \eqref{eq:equilibriumconditions}, and the second, by \eqref{eq:gsmooth}. Rearranging \eqref{eq:gsmoothinequality}, one gets the desired result.
\end{proof} 
We use the name \textit{generalized smoothness} as this novel notion of smoothness reduces to traditional smoothness when $\sum_{i=1}^n J_i(a) = C(a)$.
Observe that generalized smoothness does not even require $\sum_{i=1}^n J_i(a) \geq C(a)$, and thus applies to a much broader class of games. In parallel to the previous section, we define the \textit{generalized price-of-anarchy} as the best price-of-anarchy bound that can be derived using the generalized smoothness framework, 
\begin{equation}
    \rm{GPoA}(\mathcal{G}) := \inf_{{\lambda} > 0,{\mu} < 1} \left\{\frac{\lambda}{1-\mu} \text{ s.t. \eqref{eq:gsmooth} holds } \forall G \in \mathcal{G}\right\}. 
\end{equation}

In the following theorem, we demonstrate that the bounds obtained using the generalized smoothness framework are always better than those provided by traditional smoothness.

\begin{theorem} \label{lem:gsmoothbetter}
    For all games $G \in \mathcal{G}$ s.t. $\sum_{i=1}^n J_i(a) \geq C(a)$,
    \[ \rm{PoA}(G) \leq \rm{GPoA}(G) \leq \rm{RPoA}(G). \]
    Additionally, if for all $a \in \mathcal{A}$, $\sum_{i=1}^n J_i(a) > C(a)$. Then,
    \[ \rm{GPoA}(G) < \rm{RPoA}(G). \]
\end{theorem}

Since the result in \cref{lem:gsmoothbetter} holds for every game in the class $\mathcal{G}$, the inequalities hold with $\mathcal{G}$ in the place of $G$, i.e. for the whole class.

%% file: tightness.tex
\section{Local resource allocation games}
In this section, we introduce the specialized class of \textit{local resource allocation games}.
We then show that the generalized smoothness framework provides concrete and tight bounds on the price-of-anarchy relative to this class.
This analysis will extend the applicability of the linear programming approach presented in \cite{chandan2019optimal} to all coarse-correlated equilibria and to multiple resource types.

Consider a game $G$ with agent set $N = \{1, \dots, n\}$, and a finite set of resources $\mathcal{R}$, where every resource $r \in \mathcal{R}$ has a cost function $c_r:N \to \mathbb{R}$, and a cost-generating function $f_r:N \to \mathbb{R}$. Each agent $i \in N$ is associated with an action set $\mathcal{A}_i \subseteq 2^\mathcal{R}$. 
For a given allocation $a \in \mathcal{A}$, we define the system cost and local cost functions as,
\begin{align*}
    C(a) &= \sum_{r \in \mathcal{R}} c_r(|a|_r), \\
    J_i(a_i, a_{-i}) &= \sum_{r \in a_i} f_r(|a|_r), 
\end{align*}
where $|a|_r$ is the number of agents selecting resource $r$ in allocation $a$. We adopt the convention that $c_r(0) = 0$ for all $r \in \mathcal{R}$, without loss of generality.
We identify the aforementioned game with the tuple $G = (n, \mathcal{R}, \mathcal{A}, \{(c_r, f_r)\}_{r \in \mathcal{R}})$. 

We define a \textit{scalable} class of local resource allocation games $\mathcal{G}_T^n$ as the set of all $n$-player local resource allocation games in which, for every resource $r \in \mathcal{R}$, there exists $v_r \geq 0$ such that $c_r(\cdot) = v_r \cdot c(\cdot)$ and $f_r(\cdot) = v_r f(\cdot)$. The pair of functions $(c,f)$ is drawn from a finite set of resource types $T = \{(c_1,f_1), \dots, (c_m, f_m)\}$.
We refer to the functions $\{f_t\}_{t = 1}^{m}$ as \textit{distribution rules}, since each function $f_t$ describes how the value $v_r$ of its corresponding resource is split among the agents.

We observe that many classes of problems studied in the literature can be analyzed using this model. Important examples include vehicle-target assignment problems \cite{arslan2007autonomous}, set covering problems \cite{gairing2009covering,ramaswamy2017impact}, and atomic congestion games \cite{christodoulou2005price,awerbuch2005price,aland2006exact}. 
Before presenting our results, we demonstrate the generality of the local resource allocation problem formulation in the next subsection, using the example of atomic congestion games. 

\subsection{An Illustrative Example: Atomic Congestion Games} \label{ss:example}
To demonstrate the generality of the local resource allocation problem presented above, we analyze congestion games, a classical cost-minimization problem \cite{rosenthal1973class}. 
A congestion game consists of a player set $N = \{1, \dots, n\}$, and a finite set of edges $E$, where every player $i \in N$ selects a path $a_i$ from its corresponding set of paths $\mathcal{A}_i \subseteq 2^E$. Each edge $e \in E$ is associated with a latency function $\ell_e : N \to \mathbb{R}$. For a given allocation $a = (a_1, \dots, a_n)$, the system cost and local cost functions are defined as 
\begin{align*}
    C(a) &= \sum_{e \in E} \ell_e(|a|_e)\,|a|_e, \\
    J_i(a_i, a_{-i}) &= \sum_{e \in a_i} \ell_e(|a|_e).
\end{align*}
Observe that, in this context, a congestion game can be modelled as a local resource allocation game, where the set of edges corresponds to the set of resources, i.e. $\mathcal{R} = E$, the latency functions $\ell_r$ play the role of distribution rules $f_r$, and $c_r$ is substituted with $\ell_r(x) \cdot x$ for every $r \in \mathcal{R}$.

\vspace*{2mm}\noindent\textit{Congestion games with affine latencies \cite{awerbuch2005price}.} 
A special class of atomic congestion games is that of affine congestion games, in which the edge latency functions are restricted to the form $\ell_e(x) = a_e x + b_e$, where $a_e, b_e \geq 0$ for all $e \in E$. The class of affine congestion games is equivalent to the class of local resource allocation problems with two resource types; $T = \{(x^2, x),(x, 1)\}$.\footnote{Informally, this means that there are two edge types in the congestion game, those that impose a latency proportional to the number of agents selecting them, and those that have constant latency.} 

As an elementary example, consider an $n$-player game with resources $\mathcal{R} = \{r_j\}_{j=1}^4$, each with value $v_j \geq 0$. The resources $r_1$ and $r_3$ are associated with type $(x^2, x)$, whereas $r_2$ and $r_4$ have type $(x, 1)$.
For every agent $i \in N$, we define the action set $\mathcal{A}_i = \{a_1, a_2\}$ with $a_1 = (r_1,r_2)$ and $a_2 = (r_3,r_4)$. Observe that this game can be represented by the two-link network shown in \cref{fig:congestiongame}, where $J_i(a) = v_1|a|_{r_1}+v_2$ for an agent $i$ that selects action $(r_1,r_2)$, and $J_i(a) = v_3|a|_{r_3}+v_4$ for agents selecting action $(r_3,r_4)$.

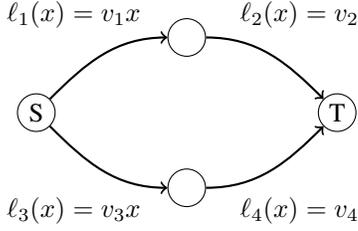
\begin{figure}[t]
    \centering
    \begin{tikzpicture}
        \node (S) [circle,draw,inner sep=0pt,minimum width=0.5cm]  at (0.0,0) {S};
        \node (T) [circle,draw,inner sep=0pt,minimum width=0.5cm] at (4.0,0) {T};
        \node (A) [circle,draw,inner sep=0pt,minimum width=0.5cm] at (2.0,1.0) {};
        \node (B) [circle,draw,inner sep=0pt,minimum width=0.5cm] at (2.0,-1.0) {};
        \node () [align=left, above] at (0.5,1.0) {$\ell_1(x) = v_1 x$};
        \node () [align=right, above] at (3.5,1.0) {$\ell_2(x) = v_2$};
        \node () [align=left, below] at (0.5,-1.0) {$\ell_3(x) = v_3 x$};
        \node () [align=right, below] at (3.5,-1.0) {$\ell_4(x) = v_4$};
        \path [black,thick,->,out=45,in=180] (S) edge (A);
        \path [black,thick,->,out=-45,in=-180]   (S) edge (B);
        \path [black,thick,->,out=0,in=135] (A) edge (T);
        \path [black,thick,->,out=0,in=-135]   (B) edge (T);
    \end{tikzpicture}
    \caption{A simple congestion game with affine latency functions that can be represented as a local resource allocation problem with two types, $T = \{(x^2, x),(x, 1)\}$. The system's $n$ agents must select either the top path or the bottom path to travel from node $S$ to node $T$, and experience the corresponding latency.}
    \label{fig:congestiongame}
\end{figure}

While we consider the simplistic example of a two-link network here, we note that, in general, any affine congestion game can be represented as a local resource allocation game with $T$ as above. 
Furthermore, given a basis set for all possible edge latency functions, \textit{any} atomic congestion game can be formulated using our model.

%% file: lp.tex
\subsection{Computing the price-of-anarchy}
The next theorem shows how the price-of-anarchy of a scalable class of local resource allocation games can be recovered by means of the notion of generalized smoothness previously defined in \eqref{eq:gsmooth}. Before proceeding, we introduce some notation. Let
\[
\begin{split}
    \mathcal{I} &\!:=\! \{(x,y,z) \!\in\! \mathbb{N}^3 \,|\, 1 \leq x+y-z \leq n,\> z \leq \min\{x,y\}\}, \\
    \mathcal{I_R} &\!:=\! \{(x,y,z) \!\in\! \mathcal{I} \,|\, x + y - z \!=\! n \text{ or } (x - z)(y - z)z \!=\! 0\}.
\end{split}
\]

\begin{theorem} \label{thm:worstcasegame}
    Given set of types $T$, and positive integer $n$, it holds that $\rm{PoA}(\mathcal{G}_T^n) = \rm{GPoA}(\mathcal{G}_T^n)$.
    Furthermore, there exists an $n$-player game $G \in \mathcal{G}_T^n$ with $|\mathcal{R}| \leq 2n$, and $\rm{PoA}(G) = \rm{GPoA}(\mathcal{G}_T^n)$.
\end{theorem}

The above theorem shows that generalized smoothness arguments provide tight upper-bounds on the price-of-anarchy in local resource allocation games, and proposes a methodology for constructing worst-case instances. We now exploit this result to obtain easily computable and concrete bounds on the price-of-anarchy.
\begin{theorem} \label{thm:equivalencetolp}
    Given set of types $T$, and positive integer $n$, $\rm{PoA}(\mathcal{G}_T^n) = 1/C^*$, where $C^*$ is the value of the following linear program,
    \begin{align} \label{eq:multiple_pairs}
            C^* = & \max_{\nu \in \mathbb{R}_{\geq 0}, \rho \in \mathbb{R}}\rho \\
            \text{s.t. } & c(y) - \rho c(x)+\nu\left[(x{-}z)f(x)-(y{-}z)f(x{+}1)\right] \geq 0 \nonumber \\
            & \hspace*{100pt} \forall (c, f) \in T \text{, } \forall (x,y,z) \in \mathcal{I_R}, \nonumber
    \end{align}
    where we set $c(0)=f(0)=f(n+1)=0$, for all $(c, f) \in T$.
\end{theorem}
%
Although all $(c, f) \in T$ were previously defined as mappings from $N\to\mathbb{R}$, we extend their definitions to ease the notation.
Note that the result in the above theorem can be used to derive exact price-of-anarchy bounds in, e.g., atomic congestion games, see \cref{tab:impossibilityresult} in \cref{sec:illustrative_example}.

\subsection{Optimizing the price-of-anarchy}%
Whereas the previous subsection was devoted to calculating the price-of-anarchy for given set of types $T$, we now shift our focus to designing a set of distribution rules that minimize the price-of-anarchy.
\begin{theorem} \label{thm:optimaldistributionrules}
    Consider the cost functions $\{c_1, \dots, c_m\}$, and positive integer $n$. 
    An optimal set of distribution rules $\mathbf{f}_\mathrm{OPT} = \{f^*_1, \dots, f^*_m\}$ such that
    \begin{equation}
        \mathbf{f}_\mathrm{OPT} \in \argmin_{\mathbf{f} \in \mathbb{R}^{n \times m}_{\geq 0}} \emph{PoA}(\mathcal{G}_T^n),
    \end{equation}
    is given by the solutions to
    \begin{align}
            (f^*_t, &\rho^*_t) \in \argmax_{f \in \mathbb{R}^n, \rho \in \mathbb{R}} \rho \label{eq:opt_distr_rule} \\
            \text{s.t. } & c_t(y) - \rho c_t(x) + (x{-}z)f(x) - (y{-}z)f(x{+}1) \geq 0,             \nonumber \\
            & \hspace*{70pt} \forall (x,y,z) \in \mathcal{I_R}, \forall t \in \{1, \dots, m\},    \nonumber
    \end{align}
    where we set $c(0)=f(0)=f(n+1)=0$, for all $(c, f) \in T$. For the set of types $T^* = \{(c_t, f^*_t)\}_{t=1}^m$,
    \begin{equation}
        \rm{PoA}(\mathcal{G}^n_{T^*}) = \max_{t \in \{1, \dots, m\}} \frac{1}{\rho^*_t}.
    \end{equation} 
\end{theorem}
%
%
\cref{thm:optimaldistributionrules} shows that a set of optimal distribution rules can be calculated using the linear program \eqref{eq:opt_distr_rule}.
It is worth noting that for a given class of games $\mathcal{G}^n_T$ with an arbitrary set of types $T$, it is not possible, in general, to compute the price-of-anarchy as the worst price-of-anarchy over each individual pair $(c_t,f_t)$, i.e. the expression 
\[ \textrm{PoA}(\mathcal{G}_T^n) = \max_{(c,f) \in T}\left\{\rm{PoA}\left(\mathcal{G}_{\{(c,f)\}}^n\right)\right\} \] 
does \textit{not hold}. Nevertheless, this property is recovered for the specific choice of $f_t=f^*_t$. This constitutes the key observation towards proving \cref{thm:optimaldistributionrules}.
\subsection{Returning to Atomic Congestion Games} \label{sec:illustrative_example}%
Here we apply the results presented in this section to the class of congestion games, as discussed in \cref{ss:example}.

\vspace*{2mm}\noindent\textit{Characterizing PoA in congestion games.} 
Deriving the smoothness parameters for a given class of congestion games is difficult. 
For example, in the main result of \cite{christodoulou2005price}, the authors exploit a nontrivial polynomial inequality in order to find the optimal smoothness parameters for the class of affine congestion games, and prove that the price-of-anarchy of this class is $5/2$. 
A direct application of the linear program in \cref{thm:equivalencetolp} recovers the same result for any number of agents greater than 3. Additionally, we determine a worst-case instance construction with only three agents\footnote{%
Consider the game $G$ with six edges $\{e_i\}_{i=1}^6$ with identical value (i.e. $v_e = v$) and latency function $\ell_e(x) = x$. We endow the $n=3$ agents with the action sets, $\mathcal{A}_1 = \{(e_4,e_5,e_6), (e_1,e_2)\}$, $\mathcal{A}_2 = \{(e_1,e_2,e_5), (e_3,e_4)\}$, and $\mathcal{A}_3 = \{(e_1,e_3,e_4), (e_5,e_6)\}$.
The Nash equilibrium $a^\mathrm{ne}$ corresponds to each agent selecting its three-tuple action, and the optimal actions in $a^\mathrm{opt}$ are the two-tuple actions. 
It can easily be verified that \eqref{eq:equilibriumconditions} is met. $\rm{PoA}(G) = 5/2$, since the system costs are $C(a^\mathrm{ne}) = 15v$ and $C(a^\mathrm{opt}) = 6v$. 
Note that, in general, drawing a worst-case instance as a graph requires additional edges $e$ with value $v_e = 0$.}.
In \cref{tab:impossibilityresult}, we compile price-of-anarchy bounds obtained using the linear program in \cref{thm:equivalencetolp} for five classes of atomic congestion games, along with their corresponding optimal smoothness parameters. 
We note that, while the price-of-anarchy bounds that we present for the first three classes of congestion games (i.e. affine, quadratic, and cubic) have already been obtained (see \cite{christodoulou2005price,awerbuch2005price,aland2006exact}), the bounds reported for classes of square root, and logarithmic congestion games are novel.
\begin{table}[t]
    \centering
    \vspace*{4pt}
    \caption{Exact price-of-anarchy bounds, and optimal smoothness parameters, for five classes of atomic congestion games. These were computed using the linear program in \cref{thm:equivalencetolp} for $n=25$. The bounds we obtain for affine, quadratic, and cubic congestion games match the tight bounds obtained for infinite player games in~\cite{christodoulou2005price,awerbuch2005price,aland2006exact}. To the best of our knowledge, we are the first to report bounds for square root, and logarithmic congestion games.}
    \label{tab:impossibilityresult}
    \begin{tabular}{|c|c|c|c|c|c|}
        \hline
        Class & Basis & $\lambda^*$ & $\mu^*$ & $\rm{PoA}$ & Reference \\
        \hline
        Affine & $\{x,1\}$ & $5/3$ & $1/3$ & $5/2$ & \cite{christodoulou2005price,awerbuch2005price} \\
        Quadratic & $\{x^2,x,1\}$ & $6.05$ & $0.368$ & $9.58$ & \cite{aland2006exact} \\
        Cubic & $\{x^3,\dots,1\}$ & $17.89$ & $0.569$ & $41.5$ & \cite{aland2006exact} \\
        Square Root & $\{\sqrt{x}\}$ & $1.24$ & $0.174$ & $1.50$ & -- \\
        Logarithmic & $\{\log(x) {+} 1\}$ & $1.523$ & $0.17$ & $1.835$ & -- \\
        \hline
    \end{tabular}
\end{table}

\vspace*{2mm}\noindent\textit{Optimizing PoA in congestion games.}
The idea of improving the price-of-anarchy in congestion games using a local edge toll as a control mechanism originates from \cite{yang2005mathematical}. 
In order to influence agents' decisions, a toll $\tau_e$ is added to the agents' local edge costs such that,
\[ J_i(a_i,a_{-i}) = \sum_{e \in a_i} \ell_e(|a|_e) + \tau_e(|a|_e). \]
The system cost remains unchanged. 
When traffic in a network is modelled as a continuum of agents -- termed the \textit{nonatomic} setting -- it has been shown that there is a unique system equilibrium achieving $\rm{PoA}(\mathcal{G}) = 1$ when agents are optimally tolled~\cite{pigou1920economics}. 
In the setting of atomic congestion games, where the agents are modelled as finite, indivisible entities, the problem of minimizing the price-of-anarchy using tolls becomes much more challenging, as a given game will have multiple equilibria, in general.
The approach presented here can be used to design a set of tolls $\tau_e$ that minimize the price-of-anarchy, see \cite{chandan2019computing}.

%% file: extensions.tex
\section{Extensions}
In this section, we demonstrate that the above results extend to \textit{A. Coarse-Correlated Equilibria}; and, \textit{B. Welfare-Maximization Problems}.

\subsection{Coarse-Correlated Equilibria} \label{sec:cce}
A significant advantage of using a smoothness argument is that it provides performance bounds for the class of \textit{coarse-correlated equilibria}, a far broader class of equilibria compared to the class of pure Nash equilibria \cite{roughgarden2009intrinsic,moulin1978strategically}. A coarse-correlated equilibrium is a probability distribution $\sigma$ over all allocations $a \in \mathcal{A}$ such that for all $i \in [n]$, and $a' \in \mathcal{A}$, it holds that,
\[\mathbb{E}_{a \sim \sigma}[J_i(a)] = \sum_{a \in \mathcal{A}} \sigma(a) J_i(a) \leq \mathbb{E}_{a\sim\sigma}[J_i(a_i', a_{-i})], \]
where $\sigma(a) \in [0,1]$ is the probability associated with action $a \in \mathcal{A}$ in the distribution $\sigma$.
Next, we show that the price-of-anarchy bounds stemming from generalized smoothness arguments extend to all coarse-correlated equilibria.
 
\begin{lemma} \label{thm:gsmoothcce}
For every game $G$ in the class of games $\mathcal{G}$,
\[ \rm{GPoA}(G) \geq \frac{\max_{\sigma \in \rm{CCE}(G)} \mathbb{E}_{a \sim \sigma} [C(a)]}{\min_{a \in \mathcal{A}} C(a)}, \]
where $\rm{CCE}(G)$ is the set of all coarse-correlated equilibria of the game $G$.
\end{lemma}

Since the sets of pure and mixed Nash equilibria of a game are subsets of its coarse-correlated equilibria, the $\textrm{GPoA}$ is an upper-bound on the efficiency of \textit{all} equilibria within these classes. 
This result is particularly important toward the tractability of the final algorithm. Indeed, although finding a pure Nash equilibrium can be intractable, coarse-correlated equilibria can often be computed in polynomial time \cite{papadimitriou2008computing}. 

\subsection{Welfare-Maximization Problems} \label{sec:wm}
Welfare-maximization problems consist of an agent set $N$, where each agent $i \in N$ is associated with a finite action set $\mathcal{A}_i$. The global objective in to maximize the system welfare function $W:\mathcal{A} \to \mathbb{R}$, i.e. we wish to find the allocation $a^{\textrm{opt}} \in \argmax_{a \in \mathcal{A}} W(a)$. As in the previous sections, we consider a game-theoretic model where each agent $i \in N$ is associated with a local utility function $U_i : \mathcal{A} \to \mathbb{R}$ which it uses to evaluate its own actions. We represent a welfare-maximization game with a tuple $G = (N, \mathcal{A}, W, \{U_i\})$. 

Given a game $G$, a Nash equilibrium is defined as any allocation $a^\textrm{ne} \in \mathcal{A}$ such that $U_i(a^\textrm{ne}) \geq U_i(a_i, a_{-i}^\textrm{ne})$ for all $a_i \in \mathcal{A}_i$ and all $i \in N$.
The price-of-anarchy in welfare-maximization games is defined as
\[ \rm{PoA}(G) = \frac{\min_{a \in \rm{NE}(G)} W(a)}{\max_{a \in \mathcal{A}} W(a)}, \quad \rm{PoA}(\mathcal{G}) = \inf_{G \in \mathcal{G}} \rm{PoA}(G).\]
Note that according to this definition, the price-of-anarchy in welfare-maximization games is bounded from below by $0$, and from above by $1$.

Although the focus of this paper is on cost-minimization games, we note that analogues to all of our results can be derived in the context of welfare-maximization games with minor modification to the generalized smoothness condition.
A welfare-maximization game $G$ is ($\lambda,\mu$)-generalized smooth if, for all allocations $a, a' \in \mathcal{A}$, it holds that,
\[ \sum_{i \in [n]} U_i(a_i', a_{-i}) - \sum_{i \in [n]} U_i(a) + W(a) \geq \lambda W(a') - \mu W(a). \]
The price-of-anarchy of a ($\lambda,\mu$)-generalized smooth game is \textit{lower-bounded} by $\lambda/(1+\mu)$, for this class of problems.


%% file: appendix.tex
\section{Appendix}

\noindent\textbf{Proof of \cref{lem:gsmoothbetter}}
\begin{proof}
    When $\sum_{i \in N} J_i(a) = C(a)$, \eqref{eq:gsmooth} is equivalent to \eqref{eq:smoothnessconditions}. 
    When $\sum_{i \in N} J_i(a) > C(a)$, for all $\lambda, \mu$ satisfying \eqref{eq:smoothnessconditions}, the following must hold for all $a, a^* \in \mathcal{A}$,
    \[ \sum_{i \in [n]}\!\! J_i(a_i^*,a_{-i})\! -\!\!\! \sum_{i \in [n]}\!\! J_i(a)\! +\! C(a) < \lambda C(a^*) \!+\! \mu C(a). \]
    Thus, there must exist some $\epsilon > 0$ such that \eqref{eq:gsmooth} holds for $\bar\lambda = \lambda^* - \epsilon$ or $\bar\mu = \mu^* - \epsilon$, where $\lambda^*, \mu^*$ optimize \eqref{eq:rpoa}. Since $\lambda/(1-\mu)$ is increasing in both $\lambda$ and $\mu$, $\textrm{GPoA} < \lambda^*/(1-\mu^*) = \textrm{RPoA}$.
\end{proof}


\noindent\textbf{Preliminaries to the proof of \cref{thm:worstcasegame}}

\begin{definition} \label{def:setS}
$\mathcal{S}(\mathcal{G}_T^n)$ is the set of parameters $\lambda > 0$, $\mu < 1$ such that, for all $(c^t, f^t) \in T$ and all $(x,y,z) \in \mathcal{I_R}$,
\begin{equation} \label{eq:ind_cons}
    (z-x)f^t(x) + (y-z)f^t(x+1) + c^t(x) \leq \lambda c^t(y) + \mu c^t(x).
\end{equation}
\end{definition}

\begin{definition}
$\gamma (\mathcal{G}_T^n)$ is defined as,
\begin{equation}
    \gamma (\mathcal{G}_T^n) := \inf_{\lambda,\mu} \left\{\frac{\lambda}{1-\mu} : (\lambda, \mu) \in \mathcal{S}(\mathcal{G}_T^n)\right\} 
\end{equation}
\end{definition}

\begin{lemma} \label{lem:gammaisstricter}
    For the given class of games $\mathcal{G}_T^n$,
    \[ \gamma(\mathcal{G}_T^n) \geq \emph{GPoA}(\mathcal{G}_T^n). \]
\end{lemma}
\begin{proof}
Let $|a^{\textrm{ne}}| = \{x_1, \dots, x_m\}$, and $|a^{\textrm{opt}}| = \{y_1, \dots, y_m\}$. We define $z_r$ as the number of agents that select resource $r$ in both $a^{\textrm{ne}}$ and $a^{\textrm{opt}}$,
\[ z_r := |\{i \in N: r \in a_i^{\textrm{ne}}\} \cap \{i \in N: r \in a_i^{\textrm{opt}}\}| \]
where $z_r \leq \min\{x_r, y_r\}$, and $1 \leq x_r + y_r - z_r \leq n$.

The following simplification, adapted from \cite{marden2014generalized}, is instrumental in our proof of tightness,
\begin{align}
    &\sum_{i \in [n]} J_i(a_i^{\textrm{opt}}, a_{-i}^{\textrm{ne}}) - \sum_{i \in [n]} J_i(a^{\textrm{ne}}) + C(a^\textrm{ne})              \label{eq:additive_nature}\\
    = &\sum_{r \in \mathcal{R}}\! \left[ z_rf_r(x_r) + (y_r-z_r)f_r(x_r+1) \right] \nonumber\\
    & - \>\sum_{r \in \mathcal{R}}\!x_rf_r(x_r) + \sum_{r\in\mathcal{R}}\!c_r(x_r)  \nonumber\\
    = & \sum_{r \in \mathcal{R}}\! \left[ (z_r{-}x_r)f_r(x_r){+} (y_r{-}z_r)f_r(x_r{+}1){+}c_r(x_r) \right]\text{.}  \nonumber
\end{align}
We have shown that \eqref{eq:additive_nature} can be represented as a sum over a subset of the left-hand side expressions in \eqref{eq:ind_cons} corresponding to the resources in $\mathcal{R}$ weighted by their values. 
For the proof that it is sufficient to consider $(x,y,z) \in \mathcal{I_R}$, see the second part of the proof of \cite[Thm. 2]{chandan2019optimal}, and note that $(x,y,z)$ in this paper are equivalent to $(j,l,x)$ in their notation.
Thus, the parameters $(\lambda, \mu) \in \mathcal{S}(\mathcal{G}_T^n)$ are sure to satisfy the constraint in \eqref{eq:gsmooth}. 
This is because $C(a^{\textrm{ne}})$ is guaranteed to be less than or equal to \eqref{eq:additive_nature}. This implies that $\gamma(\mathcal{G}_T^n) \geq \rm{GPoA}(G_T^n)$.
\end{proof}

\begin{lemma} \label{lem:optimalityparameters}
    Consider the class of games $\mathcal{G}_T^n$. Suppose there exist $(\hat{\lambda},\hat{\mu}) \in \mathcal{S}(\mathcal{G}_T^n)$ such that,
    \[ \frac{\hat{\lambda}}{1-\hat{\mu}} = \gamma(\mathcal{G}_T^n). \]
    Then, there must be $(c^1,f^1)$, $(c^2,f^2)$ in $T$, $(x_1,y_1,z_1)$, $(x_2,y_2,z_2)$ in $\mathcal{I_R}$, and $\eta \in [0,1]$ such that,
    \begin{equation}
        \begin{split}
            &(z_j-x_j)f^j(x_j) + (y_j-z_j)f^j(x_j+1) + c^j(x_j)  \\
            =\> & \hat{\lambda}\,c^j(y_j)+\hat{\mu}\,c^j(x_j)
        \end{split}
    \end{equation}
    for $j=1,2$; and,
    \begin{equation}
        \begin{split}
            & \eta [ z_1\,f^1(x_1) + (y_1-z_1)\,f^1(x_1+1) ] \\
            & +\>\! (1\!-\!\eta) [z_2\,f^2(x_2) + (y_2-z_2)\,f^2(x_2+1) ] \\
            =\> &\eta\,x_1\,f^1(x_1) + (1-\eta)\,x_2\,f^2(x_2).
        \end{split}
    \end{equation}
\end{lemma}
\begin{proof}
    We define $\mathcal{H}_{c,f,x,y,z}$ as the set of $(\lambda, \mu) \in \mathbb{R}_{>0}\times\mathbb{R}_{<1}$ that satisfy, for given $c$, $f$, $x$, $y$ and $z$,
    \[ (z\!-\!x)f(x) \!+\! (y\!-\!z)f(x\!+\!1)\! +\! c(x) \leq \lambda\,c(y)\!+\!\mu\,c(x). \]
    We denote by $\delta\mathcal{H}_{c,f,x,y,z}$ the boundary of the set, i.e. the points $(\lambda, \mu)$ that satisfy the above inequality with equality.
    Some simplifications can be made for the cases when either $x = 0$ or $y = 0$. When $x = 0$ and $y > 0$, then $z = 0 = \min\{x,y\}$, and the set $\mathcal{H}_{c,f,0,y,0}$ contains all $\mu < 1$, and all $\lambda \geq f(1)y/c(y)$. When $x > 0$ and $y=0$, $z=0$ once again, the set $\mathcal{H}_{c,f,x,0,0}$ contains all $\lambda > 0$, and all $\mu \geq 1 - xf(x)/c(x)$.  For the halfplanes with $x > 0$ and $y > 0$, the boundary is,
    \begin{equation*}
        \begin{split}
            \mu =  -\frac{c(y)}{c(x)} \lambda + \frac{1}{c(x)}\Big[ &(z-x)f(x) + \\
            & (y-z)f(x+1) + c(x) \Big]\text{.}
        \end{split}
    \end{equation*}
    
    Note that finding $\gamma(\mathcal{G}_T^n)$ is equivalent to finding the point along the boundary that intersects the line with $\mu$-intercept equal to $1$ with the most negative slope. 
    Thus, we can find the optimal $(\lambda, \mu)$ by starting on the boundary at $\lambda = \max_y f(1)y/c(y)$, then following the boundary until we reach a line with $\mu$-intercept less than 1. There are three possibilities for the optimal $\lambda, \mu$: at $\lambda = \max_y f(1)y/c(y)$, at the intersection of two halfplanes with $x_1, x_2 > 0$, $y_1, y_2 > 0$, $z_1, z_2 \geq 0$, or at $\mu = 1 - \min_x xf(x)/c(x)$. 
    
    If the optimal $\lambda, \mu$ occurs at $\lambda = \max_y f(1)y/c(y)$, then $x_1 = 0$, $y_1 > 0$ and $z_1 = 0$, and the other halfplane has $\mu$-intercept less than one, and $x_2 > 0$, $y_2 \geq 0$ and $z_2 \geq 0$. 
    Note that it is also possible for the optimal $\lambda, \mu$ to occur at $\lambda = \max_{y>0} f(1)y/c(y)$ and $\mu = 1 - \min_{x>0} xf(x)/C(x)$.

    For all these cases, there exists $\eta \in [0,1]$ such that,
    \begin{align*}
        & \eta\,[ z_1\,f^1(x_1) + (y_1-z_1)\,f^1(x_1+1) ] \\
        & +\> (1-\eta)\,[z_2\,f^2(x_2) + (y_2-z_2)\,f^2(x_2+1) ] \\
        =\> & \eta x_1 f^1(x_1) + (1-\eta) x_2 f^2(x_2)
    \end{align*}
    
    If the optimal $(\lambda, \mu)$ occur on a halfplane $\mathcal{H}_{c,f,x,y,z}$ with $\mu$-intercept equal to 1, we select $c_1 = c_2 = c$, $f^1 = f^2 = f$, $x_1 = x_2 = x$, $y_1 = y_2 = y$ and $z_1 = z_2 = z$, where any $\eta \in [0,1]$ will satisfy the equality.
\end{proof}

\begin{lemma} \label{lem:inf_not_attained}
    For the class of games $\mathcal{G}_T^n$, suppose no point $(\lambda, \mu) \in \mathcal{S}(\mathcal{G}_T^n)$ satisfies $\frac{\lambda}{1 - \mu} = \gamma(\mathcal{G}_T^n)$. Then, there exists $(f,c) \in T$ and $(x,y,z) \in \mathcal{I_R}$ such that
    \begin{align}
        \gamma(\mathcal{G}_T^n) &= \frac{c(x)}{c(y)} \label{eq:inf_poa} \\
        (y-z)f(x+1) + zf(x) &> xf(x) \label{eq:inf_nash_satisfied}
    \end{align}
\end{lemma}
\begin{proof}
    Borrowing the notation and reasoning of the proof for \cref{lem:optimalityparameters}, we know that the strictest constraint must come from a line corresponding to some $(f,c) \in T$ that for some values of $x$, $y$ and $z$ has $\mu$-intercept greater than $1$, and the least negative slope among all constraints. Since the $\mu$-intercept is greater than $1$, $(z-x)f(x)+(y-z)f(x+1) > 0$, which implies that $(y-z)f(x+1) + zf(x) > xf(x)$.
    The least negative slope results from selecting $y = \argmin_{j \in N} c(j)$ and $x = \argmax_{j \in N} c(j)$. Much like in \cite[Lem. 5.5]{roughgarden2015intrinsic}, we construct a sequence $\{(\lambda_k, \mu_k)\}$ in $\mathcal{S}(\mathcal{G}_T^n)$ such that $\frac{\lambda_k}{1-\mu_k} \downarrow \gamma(\mathcal{G}_t^n)$. Since $\frac{\lambda}{1-\mu}$ is increasing in both $\lambda$ and $\mu$, it can be assumed that every point $(\lambda_k, \mu_k)$ lies on the boundary of $\mathcal{S}(\mathcal{G}_T^n)$. The values $\lambda_k$ are bounded from below by the constraints \eqref{eq:ind_cons} where $x=z=0$, and for finite $\gamma(\mathcal{G}_T^n)$, $\mu_k \leq b < 1$. Since $\frac{\lambda}{1-\mu}$ is continuous, $\frac{\lambda_k}{1-\mu_k} \downarrow \gamma(\mathcal{G}_t^n)$ and $\gamma(\mathcal{G}_t^n)$ is not attained, the sequence $\{\lambda_k, \mu_k)\}$ has no limit point. Thus, after some rearranging of \eqref{eq:ind_cons},
    \begin{align*}
        & \gamma(\mathcal{G}_T^n) = \lim_{k \to \infty} \frac{\lambda_k}{1-\mu_k} \\
        = &\lim_{k \to \infty} \frac{c(x)}{c(y)} + \frac{(z-x)f(x) + (y-z)f(x+1)}{c(y)(1-\mu_k)} = \frac{c(x)}{c(y)},
    \end{align*}
    since $\mu_k \to -\infty$, which completes the proof.
\end{proof}

\noindent\textbf{Proof of \cref{thm:worstcasegame}}
\begin{proof}
    We first consider the case where the value $\gamma(\mathcal{G}_T^n)$ is not attained for any point $(\lambda, \mu) \in \mathcal{S}(\mathcal{G}_T^n)$ as in \cref{lem:inf_not_attained}. We recover the pair $(f,c) \in T$ that result in the strictest constraint at $\lambda \to \infty$, $\mu \to -\infty$, as well as the values $x$, $y$ and $z$ that give the least negative slope. We setup a game with $l = \min\{x+y, n\}$ resources organized in a cycle and $l$ agents, i.e. $\mathcal{R} = \{r_1, \dots, r_l\}$ and $N = [l]$, where every resource has type corresponding to the pair $(f,c)$. Each agent $i \in [l]$ is endowed with two actions, the first is to select $x$ consecutive resources starting with $r_i$ and ending with $r_{i+x-1 \bmod l}$, while the second is to select $y$ consecutive resources ending with $r_{i+z-1 \bmod l}$. Condition \eqref{eq:inf_nash_satisfied} implies that the former strategy is a Nash equilibrium, and by \eqref{eq:inf_poa}, the price-of-anarchy is at least $\frac{c(x)}{c(y)} = \gamma(\mathcal{G}_T^n)$, as required.

    We retrieve $(c^1,f^1)$, $(c^2,f^2)$, $(x_1,y_1,z_1)$, $(x_2,y_2,z_2)$ and $\eta$; the optimality parameters as in \cref{lem:optimalityparameters}, where $\gamma(\mathcal{G}_T^n)$ is an upper-bound on $\textrm{GPoA}(\mathcal{G}_T^n)$ and is guaranteed to be attained, by \cref{lem:gammaisstricter,def:setS}. The worst-case game $G$ is constructed in the following way, define two disjoint cycles $E_1$ and $E_2$ each with $l = \min \{ \max \{x_1+y_1,x_2+y_2\}, n \} $ resources enumerated from $1$ to $l$. The resources in $E_1$ are assigned cost function $c_1$, distribution rule $f^1$ and value $\eta$, whereas the resources in $E_2$ are assigned $c_2$, $f^2$ and $(1-\eta)$. There are also $l \leq n$ players, enumerated $1$ through $l$ and we restrict the action set $\mathcal{A}$ to two strategies, $a^{\textrm{ne}}$ and $a^{\textrm{opt}}$. In the first strategy, each player $i \in [l]$ selects $x_1$ consecutive resources in $E_1$, $[i, i+1, \dots, i+x_1-1] \mod l$, and $x_2$ consecutive resources in $E_2$ starting with resource $i$. In the second strategy, player $i$ selects $y_1$ consecutive resources in $E_1$ ending with resource $i-1$, and $y_2$ consecutive resources in $E_2$ ending with resource $i-1$. 

    We continue by demonstrating that the first strategy satisfies the conditions for a Nash equilibrium,
    \begin{align}
        & J_i(a^{\textrm{ne}}) = \eta\,x_1\,f^1(x_1) + (1-\eta)\,x_2\,f^2(x_2) \nonumber \\ 
        &= \eta\,[ z_1\,f^1(x_1) + (y_1-z_1)\,f^1(x_1+1) ] \nonumber \\
        &\quad +\> (1-\eta)\,[z_2\,f^2(x_2) + (y_2-z_2)\,f^2(x_2+1) ] \label{eq:unilateral_deviation} \\
        &= J_i(a_i^{\textrm{opt}}, a_{-i}^{\textrm{ne}})\text{,} \nonumber
    \end{align}
    where \eqref{eq:unilateral_deviation} holds due to \cref{lem:optimalityparameters}. Now we show that the price-of-anarchy of the game is lower-bounded by $\gamma(\mathcal{G}_T^n)$, thus implying equality.
    \begin{align*}
        & C(a^{\textrm{ne}})= C(a^{\textrm{ne}}) - \sum_{i=1}^k J_i(a^{\textrm{ne}}) + \sum_{i=1}^k J_i(a_i^{\textrm{opt}}, a_{-i}^{\textrm{ne}}) \\
        &= k\,\eta\,\Big[\hat{\lambda}\,c_1(y_1)+\hat{\mu}\,c_1(x_1)\Big] \\
        & \quad +\> k\,(1-\eta)\,\Big[\hat{\lambda}\,c_2(y_2)+\hat{\mu}\,c_2(x_2)\Big] \\
        &= \hat{\lambda}\,C(a^{\textrm{opt}}) + \hat{\mu}\,C(a^{\textrm{ne}})
    \end{align*}
    In the above, $\gamma(\mathcal{G}_T^n) = \rm{PoA}(G) \leq \rm{PoA}(\mathcal{G}_T^n)$. Since $\gamma(\mathcal{G}_T^n) \geq \textrm{GPoA}(\mathcal{G}_T^n) \geq \textrm{PoA}(\mathcal{G}_T^n)$ by \cref{lem:gammaisstricter}, $\textrm{GPoA}(\mathcal{G}_T^n)$ must be tight.
\end{proof}

\noindent\textbf{Proof of \cref{thm:equivalencetolp}}
\begin{proof}
We begin by noting that, by \cref{def:setS}, we need only consider $(x,y,z) \in \mathcal{I_R}$ when calculating the price-of-anarchy in local resource allocation games. Observe that the constraints in the linear program are equivalent to the simplified conditions for $(\lambda, \mu)$-generalized smoothness in \eqref{eq:ind_cons}. The linear program constraints read as,
\[ c(y)\! -\! \rho\,c(x)\! +\! \nu\left[ (x\!-\!z)\,f(x)\! -\! (y\!-\!z)\,f(x\!+\!1) \right] \geq 0, \]
for all $(x,y,z) \in \mathcal{I}_\mathcal{R}$, where $\rho = \frac{1-\mu}{\lambda}$, and $\nu = \frac{1}{\lambda}$. Substituting the expressions for $\nu$ and $\rho$ into the above, and rearranging, we are left with,
\begin{align*}
     & (z-x)\,f(x) + (y-z)\,f(x+1) + c(x) \\
     \leq\> &\lambda\,c(y) + \mu\,c(x)\text{,}
\end{align*}
for all  $(x,y,z) \in \mathcal{I}_\mathcal{R}$, which is identical to \eqref{eq:ind_cons} when there is a single type. 
Next, observe that maximizing $\rho$ is equivalent to minimizing $\lambda/(1-\mu)$, which concludes the proof.
\end{proof}

\begin{lemma} \label{lem:poagreaterthanequal}
    For a given class of local resource allocation games $\mathcal{G}_T^n$, it holds that,
    \begin{equation} \label{eq:poageq}
        \emph{PoA}(\mathcal{G}_T^n) \geq \max_{\mathbf{t} \in T}\left\{\emph{PoA}\left(\mathcal{G}_\mathbf{t}^n\right)\right\} \text{.}
    \end{equation}
\end{lemma}
\begin{proof}
    We begin by proving that it is impossible to have $\textrm{PoA}(\mathcal{G}_T^n) < \max_{\mathbf{t}\in T} \{\textrm{PoA}(\mathcal{G}_\mathbf{t}^n)\}$. Simply note that the worst-case game in $\mathcal{G}_\mathbf{t}^n$ for each $\mathbf{t} \in T$ is a member of the class of games $\mathcal{G}_T^n$.
    
    Next, consider the class of games with $n = 3$, and $T = \{T_1, T_2\} = \{(x^2, x), (x,x)\}$. By \cite[Thm. 2]{chandan2019optimal}, the prices-of-anarchy for the games with the individual types are $\textrm{PoA}(\mathcal{G}_{T_1}^3) = 1.857$ and $\textrm{PoA}(\mathcal{G}_{T_2}^3) = 2.0$. But, $\rm{PoA}(\mathcal{G}_T^n) = 2.6$ by \eqref{eq:multiple_pairs}.
\end{proof}

\begin{lemma} \label{lem:scalingparameters}
    For a given class of local resource allocation games $\mathcal{G}_T^n$, there exist scaling parameters $\alpha_t \in \mathbb{R}_{\geq 0}$, $t \in |T|$ such that,
    \[ \emph{PoA}(\mathcal{G}_\tau^n, n) = \max_{\mathbf{t}\in T} \{\emph{PoA}(\mathcal{G}_\mathbf{t}^n)\}, \]
    where $\tau = \{(c^t, \alpha_tf^t)\}_{t=1}^{|T|}$.
\end{lemma}
\begin{proof}
    We denote by $(\nu^*_t, \rho^*_t)$ the solution to \cite[Thm. 2]{chandan2019optimal} for the class of games with one type, $(c^t,f^t) \in T$. First, note that uniform scaling of the distribution rules does not affect the equilibrium conditions, so $\textrm{PoA}(\{(c^t, \alpha_tf^t)\}, n) = \textrm{PoA}(\{(c^t, f^t)\}, n)$ for all $\alpha_t > 0$ and all $\mathbf{t} \in T$. Thus, recalling \cref{lem:poagreaterthanequal},
    \[ \textrm{PoA}(\mathcal{G}_\tau^n) \geq \max_{\mathbf{t}\in T} \{\textrm{PoA}(\mathcal{G}_\mathbf{t}^n)\}, \]
    where $\tau = \{(c^t, \alpha_1f^t)\}_{t=1}^{|T|}$. Select $\alpha_t = \nu^*_t$ for all $t \in [|T|]$, such that $\tau := \{(c^t, \nu^*_tf^t)\}_{t=1}^{|T|}$.
    We define $\hat{\rho} := \min_{t\in [|T|]} \rho_t^*$. By construction, $(\hat{\rho}, 1)$ satisfies all the constraints in \eqref{eq:multiple_pairs} for types in $\tau$. Thus, $\textrm{PoA}(\mathcal{G}_\tau^n) \leq 1/\hat{\rho} = \max_{\mathbf{t}\in T} \{\textrm{PoA}(\mathcal{G}_\mathbf{t}^n)\}$.
\end{proof}

\noindent\textbf{Proof of \cref{thm:optimaldistributionrules}}
\begin{proof}
    By \cref{lem:poagreaterthanequal}, the lowest achievable price-of-anarchy is $\max_{\mathbf{t} \in T^*} \{\textrm{PoA}(\mathcal{G}_\mathbf{t}^n)\}$ where $T^* := \{(c^t, f_\textrm{OPT}^t)\}_{t=1}^{|T|}$. Additionally, each of the $f_\textrm{OPT}^t$ minimizes its corresponding $\textrm{PoA}(\mathcal{G}_\mathbf{t}^n)$ by \cite[Thm. 3]{chandan2019optimal}. Finally, we have that the following statement,
    \[ \textrm{PoA}(\mathcal{G}_{T^*}^n) = \max_{\mathbf{t}\in T} \{\textrm{PoA}(\mathcal{G}_\mathbf{t}^n)\}, \]
    holds by the construction of $f_{\textrm{OPT},t}$ in \cite[Thm. 3]{chandan2019optimal}; the linear program already multiplies the distribution rule and $\lambda_t^*$, and it was shown in the proof of \cref{lem:scalingparameters} that $\alpha_t = \lambda_t^*$ for all $t \in T$ is an optimal set of scaling parameters.
\end{proof}

%% file: main.bbl
\begin{thebibliography}{10}
\providecommand{\url}[1]{#1}
\csname url@samestyle\endcsname
\providecommand{\newblock}{\relax}
\providecommand{\bibinfo}[2]{#2}
\providecommand{\BIBentrySTDinterwordspacing}{\spaceskip=0pt\relax}
\providecommand{\BIBentryALTinterwordstretchfactor}{4}
\providecommand{\BIBentryALTinterwordspacing}{\spaceskip=\fontdimen2\font plus
\BIBentryALTinterwordstretchfactor\fontdimen3\font minus
  \fontdimen4\font\relax}
\providecommand{\BIBforeignlanguage}[2]{{%
\expandafter\ifx\csname l@#1\endcsname\relax
\typeout{** WARNING: IEEEtran.bst: No hyphenation pattern has been}%
\typeout{** loaded for the language `#1'. Using the pattern for}%
\typeout{** the default language instead.}%
\else
\language=\csname l@#1\endcsname
\fi
#2}}
\providecommand{\BIBdecl}{\relax}
\BIBdecl

\bibitem{bussmann2013multiagent}
S.~Bussmann, N.~R. Jennings, and M.~Wooldridge, \emph{Multiagent systems for
  manufacturing control: a design methodology}.\hskip 1em plus 0.5em minus
  0.4em\relax Springer Science \& Business Media, 2013.

\bibitem{wu2017distributed}
J.~Wu, \emph{Distributed system design}.\hskip 1em plus 0.5em minus 0.4em\relax
  CRC press, 2017.

\bibitem{yu2016smart}
X.~Yu and Y.~Xue, ``Smart grids: A cyber--physical systems perspective,''
  \emph{Proceedings of the IEEE}, 2016.

\bibitem{spieser2014toward}
K.~Spieser, K.~Treleaven, R.~Zhang, E.~Frazzoli, D.~Morton, and M.~Pavone,
  ``Toward a systematic approach to the design and evaluation of automated
  mobility-on-demand systems: A case study in singapore,'' in \emph{Road
  vehicle automation}.\hskip 1em plus 0.5em minus 0.4em\relax Springer, 2014.

\bibitem{le2015decentralized}
T.~Le, P.~Kov{\'a}cs, N.~Walton, H.~L. Vu, L.~L. Andrew, and S.~S. Hoogendoorn,
  ``Decentralized signal control for urban road networks,''
  \emph{Transportation Research Part C: Emerging Technologies}, 2015.

\bibitem{el2013distributed}
M.~I. El-Hawwary and M.~Maggiore, ``Distributed circular formation
  stabilization for dynamic unicycles,'' \emph{IEEE TAC}, 2013.

\bibitem{zhang2016collaborative}
H.-T. Zhang, Z.~Chen, and M.-C. Fan, ``Collaborative control of multivehicle
  systems in diverse motion patterns,'' \emph{IEEE Transactions on Control
  Systems Technology}, 2016.

\bibitem{nedic2009distributed}
A.~Nedic and A.~Ozdaglar, ``Distributed subgradient methods for multi-agent
  optimization,'' \emph{IEEE Transactions on Automatic Control}, 2009.

\bibitem{wei2013distributed}
E.~Wei, A.~Ozdaglar, and A.~Jadbabaie, ``A distributed newton method for
  network utility maximization--i: Algorithm,'' \emph{IEEE Transactions on
  Automatic Control}, 2013.

\bibitem{shamma2007cooperative}
J.~S. Shamma, \emph{Cooperative control of distributed multi-agent
  systems}.\hskip 1em plus 0.5em minus 0.4em\relax Wiley Online Library, 2007.

\bibitem{arslan2007autonomous}
G.~Arslan, J.~R. Marden, and J.~S. Shamma, ``Autonomous vehicle-target
  assignment: A game-theoretical formulation,'' \emph{Journal of Dynamic
  Systems, Measurement, and Control}, 2007.

\bibitem{christodoulou2005price}
G.~Christodoulou and E.~Koutsoupias, ``The price of anarchy of finite
  congestion games,'' in \emph{Proceedings of the thirty-seventh annual ACM
  symposium on Theory of computing}.\hskip 1em plus 0.5em minus 0.4em\relax
  ACM, 2005.

\bibitem{awerbuch2005price}
B.~Awerbuch, Y.~Azar, and A.~Epstein, ``The price of routing unsplittable
  flow,'' in \emph{Proceedings of the thirty-seventh annual ACM symposium on
  Theory of computing}.\hskip 1em plus 0.5em minus 0.4em\relax ACM, 2005.

\bibitem{aland2006exact}
S.~Aland, D.~Dumrauf, M.~Gairing, B.~Monien, and F.~Schoppmann, ``Exact price
  of anarchy for polynomial congestion games,'' in \emph{Annual Symposium on
  Theoretical Aspects of Computer Science}.\hskip 1em plus 0.5em minus
  0.4em\relax Springer, 2006.

\bibitem{gairing2009covering}
M.~Gairing, ``Covering games: Approximation through non-cooperation,'' in
  \emph{WINE}.\hskip 1em plus 0.5em minus 0.4em\relax Springer, 2009.

\bibitem{roughgarden2009intrinsic}
T.~Roughgarden, ``Intrinsic robustness of the price of anarchy,'' in
  \emph{Proceedings of the forty-first annual ACM symposium on Theory of
  computing}.\hskip 1em plus 0.5em minus 0.4em\relax ACM, 2009.

\bibitem{paccagnan2018distributed}
D.~Paccagnan, R.~Chandan, and J.~R. Marden, ``Utility design for distributed
  resource allocation–part i: Characterizing and optimizing the exact price
  of anarchy,'' \emph{arXiv:1807.01333}, 2018.

\bibitem{chandan2019optimal}
R.~Chandan, D.~Paccagnan, and J.~R. Marden, ``Optimal price of anarchy in
  cost-sharing games,'' \emph{American Control Conference}, 2019.

\bibitem{caragiannis2015bounding}
I.~Caragiannis, C.~Kaklamanis, P.~Kanellopoulos, M.~Kyropoulou, B.~Lucier,
  R.~P. Leme, and {\'E}.~Tardos, ``Bounding the inefficiency of outcomes in
  generalized second price auctions,'' \emph{Journal of Economic Theory}, 2015.

\bibitem{foster2016learning}
D.~J. Foster, Z.~Li, T.~Lykouris, K.~Sridharan, and E.~Tardos, ``Learning in
  games: Robustness of fast convergence,'' in \emph{NIPS}, 2016.

\bibitem{syrgkanis2013composable}
V.~Syrgkanis and E.~Tardos, ``Composable and efficient mechanisms,'' in
  \emph{Proceedings of the 45th annual ACM symposium on Theory of
  computing}.\hskip 1em plus 0.5em minus 0.4em\relax ACM, 2013.

\bibitem{ramaswamy2017impact}
V.~Ramaswamy, D.~Paccagnan, and J.~R. Marden, ``The impact of local information
  on the performance of multiagent systems,'' \emph{arXiv:1710.01409}, 2017.

\bibitem{papadimitriou2008computing}
C.~H. Papadimitriou and T.~Roughgarden, ``Computing correlated equilibria in
  multi-player games,'' \emph{Journal of the ACM (JACM)}, 2008.

\bibitem{roughgarden2015intrinsic}
T.~Roughgarden, ``Intrinsic robustness of the price of anarchy,'' \emph{Journal
  of the ACM (JACM)}, 2015.

\bibitem{roughgarden2017price}
T.~Roughgarden, V.~Syrgkanis, and E.~Tardos, ``The price of anarchy in
  auctions,'' \emph{Journal of Artificial Intelligence Research}, 2017.

\bibitem{rosenthal1973class}
R.~W. Rosenthal, ``A class of games possessing pure-strategy nash equilibria,''
  \emph{International Journal of Game Theory}, 1973.

\bibitem{yang2005mathematical}
H.~Yang and H.-J. Huang, \emph{Mathematical and economic theory of road
  pricing}, 2005.

\bibitem{pigou1920economics}
A.~Pigou, \emph{The economics of welfare}.\hskip 1em plus 0.5em minus
  0.4em\relax Macmillan, 1920.

\bibitem{chandan2019computing}
R.~Chandan, D.~Paccagnan, B.~L. Ferguson, and J.~R. Marden, ``Computing optimal
  taxes in atomic congestion games,'' in \emph{NetEcon}, 2019.

\bibitem{moulin1978strategically}
H.~Moulin and J.-P. Vial, ``Strategically zero-sum games: the class of games
  whose completely mixed equilibria cannot be improved upon,''
  \emph{International Journal of Game Theory}, 1978.

\bibitem{marden2014generalized}
J.~R. Marden and T.~Roughgarden, ``Generalized efficiency bounds in distributed
  resource allocation,'' \emph{IEEE TAC}, 2014.

\end{thebibliography}
